\documentclass[11pt,a4paper]{article}

\usepackage{amsmath,amsfonts,amssymb,verbatim}
\usepackage{amssymb,amsthm}
\usepackage{graphicx}
\usepackage{url}
\usepackage{color,soul}
\usepackage{hyperref}
\usepackage{times}
\usepackage{microtype} 
\usepackage[T1]{fontenc} 
\usepackage{enumitem}
\usepackage{multirow}
\usepackage{tabularx}

\newtheorem{problem}{Problem}
\renewcommand{\vec}[1]{\mathbf{#1}}
\newcommand{\N}{\mathbb{N}}
\newcommand{\EF}{\mathcal{F}}
\newcommand{\x}{\vec{x}}
\newcommand{\e}{\vec{e}}
\newcommand{\ea}{\vec{e_{\sf a}}}
\newcommand{\eo}{\vec{e_{\sf o}}}
\newcommand{\et}{\vec{e_{\sf t}}}

\newcommand{\Reader}{\mathcal{R}}
\newcommand{\Energy}{\mathcal{E}}
\newcommand{\Impact}[2]{\vec{#1}^{\oplus #2}}

\newcommand{\com}[3]{\ensuremath{\mathsf{Compare}(#1,#2,#3)}}

\newcommand{\E}{\mathbb{E}}
\newcommand{\Inf}{\mathsf{Inf}}

\newcommand{\wkt}[1]{\mathsf{wkt}^{#1}}

\newcommand{\TIm}{\mathsf{TIm}}
\newcommand{\xor}{\mathsf{XOR}}

\newtheorem{theorem}{Theorem}
\newtheorem{corollary}{Corollary}
\newtheorem{lemma}{Lemma}
\newtheorem{definition}{Definition}

\newtheorem{claim}{Claim}

\title{Algorithmic Foundations of Inexact Computing}

\author{John Augustine\thanks{Department of Computer Science and Engineering, IIT Madras, Chennai, India. Email: \href{mailto: augustine@iitm.ac.in}{\tt   augustine@iitm.ac.in}} \and Dror Fried \thanks{Department of Mathematics and Computer Science, The Open University of Israel, Israel. Email: \href{dfried@openu.ac.il}{\tt   dfried@openu.ac.il}} 
\and Krishna V. Palem\thanks{Department of Computer Science, Rice University, United States of America. Email: \href{mailto: Krishna.V.Palem@rice.edu}{\tt   Krishna.V.Palem@rice.edu}}
\and Duc-Hung Pham \thanks{Department of Computer Science, Rice University, USA. Email: \href{mailto: hungdpham92@gmail.com}{\tt  hungdpham92@gmail.com}} \and
Anshumali Shrivastava\thanks{Department of Computer Science, Rice University, USA.  Email: \href{mailto: anshumali@rice.edu}{\tt  anshumali@rice.edu}}}

\begin{document}

\maketitle

\begin{abstract}

\emph{Inexact computing} also referred to as \emph{approximate computing}  is a style of designing algorithms and computing systems wherein
the accuracy of correctness of algorithms executing on them is deliberately traded for significant resource savings. 
Significant progress has been reported in this regard both in terms of hardware as well as 
software or custom algorithms that exploited this approach resulting in some loss in solution quality (accuracy) while garnering 
disproportionately high savings. However, these approaches  tended to be ad-hoc and were 
tied to specific algorithms and technologies.  Consequently, a principled approach to designing 
and analyzing algorithms was lacking.

In this paper, we provide a novel \emph{model} which allows us to characterize the behavior of algorithms designed to be inexact,
as well as characterize \emph{opportunities} and benefits that this approach offers. Our methods therefore are amenable
to standard 
asymptotic analysis and provides a clean unified abstraction through which an algorithm's
design and analysis can be conducted. With this as a backdrop, we show that inexactness can be significantly beneficial for some fundamental problems in that
the quality of a solution can be \emph{exponentially} better if one exploits inexactness when compared to approaches that are
agnostic and are unable to exploit this approach. We show that such gains are possible in the context of evaluating \emph{Boolean
functions} rooted in the theory of Boolean functions and their spectra~\cite{AnalysisBook2014}, PAC learning~\cite{Valiant}, and sorting. 
Formally, this is accomplished by introducing the twin concepts of \emph{inexactness aware} and \emph{inexactness oblivious} approaches to
designing algorithms and the exponential gains are shown in the context of taking the ratio of the quality of the solution using the ``aware''
approach to the ``oblivious'' approach. 



\end{abstract}

\section{Introduction}

Much of the impetus for increased performance and ubiquity of information technologies is derived from the exponential
rate at which technology could be miniaturized. Popularly referred to as Moore's law~\cite{moore}, this trend persisted from the 
broad introduction of integrated circuits over five decades ago, and was built on the promise of halving the size
of  transistors which are hardware building blocks roughly every eighteen months. As transistors started approaching 
10 nanometers in size, two major hurdles emerged and threatened the hitherto uninterrupted promise of Moore's law. First,
engineering reliable devices that provide a basis for viewing them as ``deterministic'' building blocks started becoming
increasingly hard. Various hurdles emerged ranging from vulnerability to noise~\cite{kish, palem-cmos-japan} to 
vulnerabilities such as ensuring reliable interconnections~\cite{borkar}. Additionally, smaller devices held out the allure that
more of them could be packed into the same area or volume thus increasing the amount of computational power that 
could be crammed into a single chip, while at the same time supporting smaller switching times implying faster clock
speeds characterized as \emph{Dennard scaling}. However, this resulted in  more switching activity within a given area causing greatly increased 
energy consumption, often referred to as the ``power wall''~\cite{borkar}, as well as heat dissipation needs. 

Motivated by these hurdles, intense research along dimensions as diverse as novel devices and materials such 
as graphene~\cite{graphene}, as well as fundamentally novel computing frameworks  including quantum~\cite{quantum1, quantum2, quantum3} and
DNA ~\cite{dna1, dna2} based approaches have been proposed. However, a common theme in all of these efforts is the need  to  preserve the 
predictable and repeatable  or deterministic behavior that the resulting computers exhibit, very much in keeping with Turing's original vision~\cite{turing}.
Faced with a similar predicament when digital computers were in their infancy and their components were notoriously unreliable, 
pioneers such as von Neumann advocated methods for realizing reliable computing from unreliable elements, achieved through
error correction~\cite{vonNeumann}. Thus, the march towards realizing computers which retain their impressive reliability continues
unabated.

In sharp contrast, \emph{inexact computing}~\cite{PL13, KP14}  was proposed as an unorthodox alternative to overcoming these hurdles, specifically by 
embracing ``unreliable'' hardware without attempting to rectify erroneous behavior. The resulting computing
architectures solve the problem where the \emph{quality} of the
solution is traded  for \emph{disproportionately} high savings in (energy) 
resource consumption. The counter-intuitive consequence of this approach was that by embracing hardware architectures that 
operate erroneously as device sizes shrink, and deliberately so~\cite{C08, chakrapani2007}, one could simultaneously garner energy savings! Thus,
by accepting less than accurate hardware as a design choice, we can simultaneously overcome the energy or power wall. 
Therefore, in the inexactness regime, devices and resulting computing architectures
are allowed to remain unreliable, and the process of deploying algorithms involves \emph{co-designing}~\cite{C08, chakrapani2007} them with the technology and
 the architecture. This resulted in the need for novel algorithmic methods  that trade off the quality (accuracy) of their
solutions for savings in cost and (energy) resource consumption.

To give this context, let us consider the behavior of a single inverter (gate) shown in  Figure~\ref{fig:energy-probability}. Here, the probability of correct operation $q$
of the gate is measured  as the energy consumed by the gate increases. It is interesting to note that the energy consumed \emph{increases} exponentially with $q$ 
given by the 
relationship.  
Suppose we have spend $e$ units of energy to inexactly read  a bit $b$. Due to the inexactness, the bit read is $b'$. The probability with which $b'$ differs from $b$ will depend on $e$. 
Modeled on empirically validated physical measurements\footnote{In its full form using 
CMOS characteristics, the probability of error is $p = \frac{1}{2} {\sf erfc}\left (\frac{V_{dd}}{2 \sqrt{2} \sigma} \right )$ where the error
function ${\sf erfc}(x) = \frac{2}{\sqrt{\pi}} \int_x^\infty e^{-u^2} du$.}, we will use the clean abstraction that the probability of error $p = (1-q) =  Pr [b \neq b'] = \frac{1}{2^e}$. Thus, a small decrease on the probability of correctness from the ``desired'' value of $1$ will result in a disproportionately
large savings in energy consumed~\cite{korkmaz2006}. The {\em inexact design philosophy} is to assign different amounts of energy (or other resources) strategically to different parts of the computation in order to achieve useful trade offs between energy and the quality of the outcome.

Building on the inexact design philosophy,  quite a few  results were published through architectural artifacts that enabled trading the accuracy or 
quality of a solution, notably for energy consumption. Early examples included specialized architectures for signal processing~\cite{george2006},   neural networks~\cite{avinash1} 
and floating point operations for weather prediction~\cite{avinash2, duben15}. The overarching template for these designs was that of a 
co-processor or a processor parts of which could be rendered inexact~\cite{chakrapani2007, palempatent}. In literature, inexact computing also goes by approximate computing. Mittal's survey~\cite{M15} and references therein (along with its many citations) are a testament to the broad impact of inexact/approximate computing. 
The approaches in general involved exposing the hardware features to
and customize the algorithm design to realize the solution by being cognizant of the architectural tradeoffs that the technology offered. This process was heuristic and
generally ad-hoc due to the lack of \emph{principled} methodologies for design and analysis of algorithms in this setting. 

In this paper, we aim to remedy this situation by providing a clean and simple
framework for exposing the unreliable aspects of underlying hardware to the algorithm design process through a foundational model,  amenable to rigorous mathematical
analysis. Intuitively, the more unreliable an element, the cheaper it is. Thus, the trade-off ubiquitous to our contribution in this paper is to strike the correct balance between cost and quality, the latter being the accuracy of the result. Thus, given a computing substrate which we model  below, we can design an algorithm and determine through rigorous analysis whether it meets the quality or accuracy needs.  Here, by rigorous analysis, we mean using asymptotic methods used by algorithm designers every day using $O (n)$ and $\Omega (n)$ where $n$ is the size 
of the input. 

To the best of our knowledge, the model we present in this paper is the first instance that offers a clean framework for algorithm design and analysis where the architectural and hardware variability is exposed thereby enabling us to leverage it for greater efficiency either in terms of speedup or energy consumption, or a suitable trade off between the two.
Many models were used in earlier works informally where researchers used heuristics to take a model of hardware and map an algorithm onto it while trading ``cost'' for ``quality'' (see \cite{kirthi-mooney}, \cite{avinahsh-tieman}, \cite{energy-parsimonious} for example, or through ad-hoc experimental methods \cite{atmospheric-modeling}, with some exceptions from earlier works in the limited domain of integer arithmetic ~\cite{kirthi-lakshmi, kirthi-lakshmi2, lakshmi-thesis} based on experimental findings~\cite{bilge}). In these contexts, the researchers were able to navigate a space of solutions, to reiterate heuristically and find a solution that provides the best ``quality'' or accuracy subject to a cost constraint or vice-versa. 
In contrast, the model we introduce here provides a mathematically tractable framework that is amenable for a principled approach to algorithm design and analysis through judiciously abstracting the parts of hardware variations that affect cost and quality. In so doing, we claim our model strikes a balance between providing an abstraction that provides adequate detail to capture the impact of hardware (cost) variations, while being simple enough for rigorous mathematical analysis.

We demonstrate the  value of our model in the context of analyzing  the effect of inexactness 
for a variety of fundamental algorithmic problems.
To lay the foundation for our work, we start with Boolean functions and basic operations like binary evaluation, XOR, etc. 
We next show the power of inexactness in the context of machine learning, a popular topic of interest, and of sorting, an important practical application. Using those functions, 
we provide a glimpse of
the spectrum of possible results and build the big picture that demonstrates the usefulness of the model. 

In the interest of eliciting the principles of inexactness, the model we present is mathematically  clean and provides an effective  abstraction for theoretical investigation of inexactness. 
In reality, the error probability of a complex operation can be calculated by breaking down that operation to computational steps and propagate through the computation. Such analyses quickly become mathematically complicated. We have therefore deliberately simplified our model of inexactness wherein we concentrate the effects of inexactness at the point where data is read. We believe that this simplification retains the principles of inexactness while dispensing with details that can be analyzed more naturally through simulation and experiments, which we hope to do in the future. 


\subsection{Related work}
\label{related}
There has been significant progress in inexact computing over the past fifteen years. Early foundational models~\cite{palem1, KP05, KP07} 
were aimed at modeling the complexity of algorithms akin to random access machines and circuits~\cite{AB09}, and are not well-suited
to support algorithm analysis and design.  Since then, much progress has been made in the context of inexact VLSI circuits and architectures (see~\cite{PL13} for
a historical perspective  and partial survey). Problem specific mathematical methods 
do exist for analyzing the effect of inexactness when 
specific problems are considered notably arithmetic~\cite{C08, LC08}, along with optimization problems through which  
cost-accuracy tradeoffs were explored~\cite{KKM11}. More recently, there has been quite a surge of interest in 
studying sorting using approximate operators but the associated models do not have 
an explicit associated cost dimension to optimize~\cite{BM08, KPSW11, AFHN15, GP18, G18}. 

\subsection{Roadmap of the paper}
In section \ref{model} and \ref{boolean} of this paper, we respectively describe our inexactness model in its full generality followed by a way of specifying Boolean functions using this model. We choose Boolean functions since they are at the core of understanding computational complexity and algorithmic behavior.  For decision problems based on evaluating Boolean functions, In section \ref{optimal-existence},  we show that  an optimal energy allocation  always exists. In section \ref{aware-oblivious} we look at the question of conditions under which  being aware of the importance of a variable in the Boolean function characterized through its influence helps. Thus influence becomes our parameter to base decisions on how an algorithm designer could make decisions about energy investments. This dichotomy is captured by the complementary notions of being ``influence aware'' versus ``influence oblivious'' approaches to algorithm design.  In section \ref{sec:learning}, we apply these insights in the context of the well-known  PAC learning \cite{Valiant} problem. Next, in section \ref{sec: sort} we study the difference between influence aware and influence oblivious approach in sorting.

\section{The general inexactness model}
\label{model}

\begin{figure}[h]
\centering
\includegraphics[width = 0.5\textwidth]{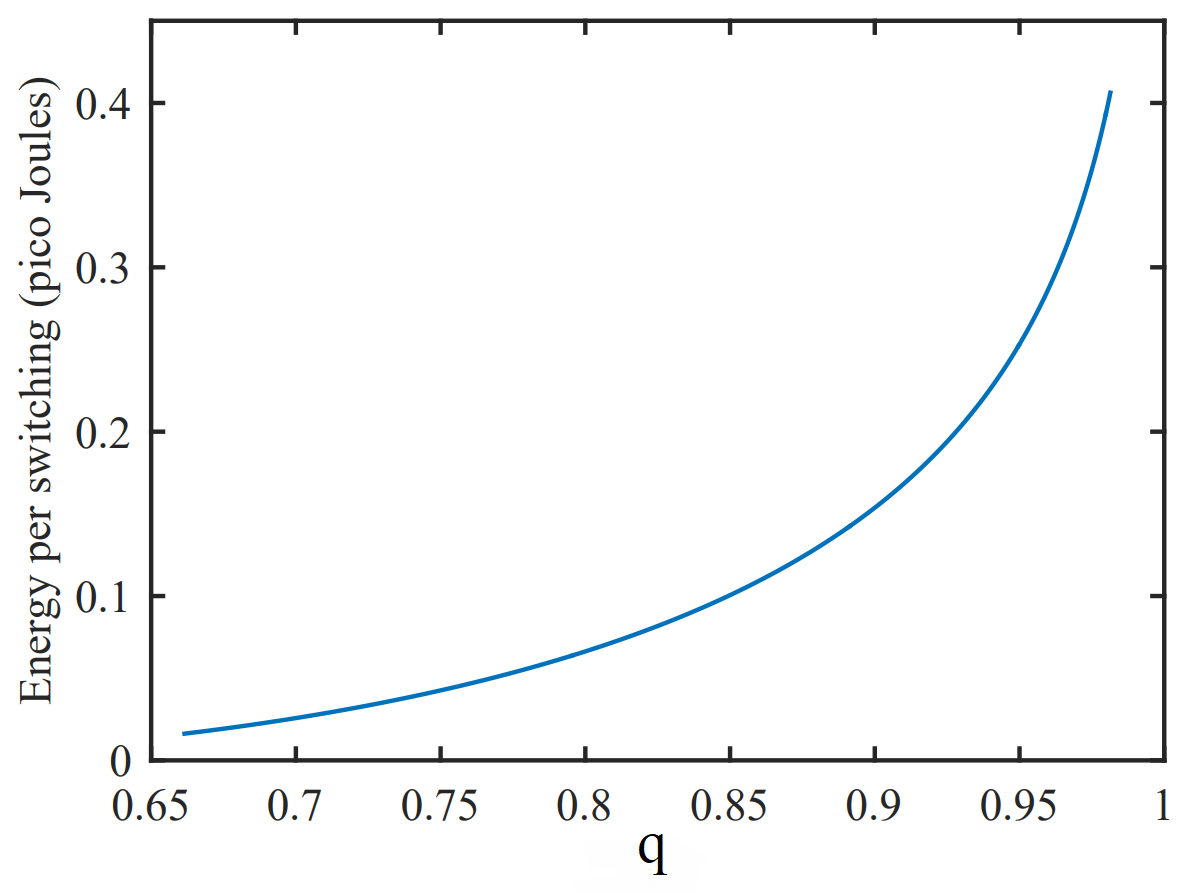}
\caption{The relationship between energy consumed and probability of correctness $q$ of a single inverter built out of CMOS technology from~\cite{korkmaz2006}}.
\label{fig:energy-probability}
\end{figure}

In inexact computing,  a   function or algorithm $f$ which could be Boolean
is computed in a noisy environment (see Figure \ref{fig:energy-probability})  where the result can be erroneous. 
To formalize this notion, 
we postulate a \textit{reader} as a function $\Reader:\{0,1\}^n\rightarrow \{0,1\}^n$ 
that ``scrambles" the data by flipping (changing from $0$ to $1$ or vice-versa) some 
of the input bits. The result of the interference of the reader is that instead of evaluating the function using ``correct'' values,  
we end up evaluating $f\circ\Reader$.

The extent to which the reader obfuscates $\x$ depends on the energy invested.
We are given an energy 
\emph{budget}  $\Energy\geq 0$ that can be apportioned  into a vector of $n$ elements 
$\vec{e}=(e_1,\cdots e_n)$ while ensuring that $\sum_{i=1}^ne_i \le  \Energy$. 
Each of the $e_i$'s determines the  probability with which our reader provides 
incorrect values of  the corresponding $\x_i$. This effect is characterized by a transformation 
$\EF:\mathbb{R}\rightarrow [0,1]$ such that  for every $i$, the reader flips  bit $i$ 
with probability $p_i=\EF(e_i)$, namely  with probability $q_i = 1-p_i$, $\x_i$ must be read correctly. 
In keeping with measured behavior of CMOS devices outlined above,  $\EF(e_i)=1/2^{e_i}$. Clearly,  the bigger $e_i$ is, 
the \textit{lower}  the chance that bit $i$ is flipped, and  $\vec{p}=(p_1,\cdots p_n)$.
Note that $\vec{p}$ is not a probability  but rather, each $p_i$  is. 

As mentioned in the previous section, this model is inspired by the behavior of physical gates such as the inverter as shown in Figure \ref{fig:energy-probability}, or NAND gate as mentioned in \cite{C08}, where an approximately exponential relationship between the error probability and the energy investment (such as energy investment in switching the value of the bit for Probability CMOS switch) was observed. The probability that an error occurs in a computational gate can be abstracted to be the probability of error of reading input bits to that gate. 
In reality, the error probability of a complex operation can be calculated by breaking down that operation to computational steps, and aggregate the error probability throughout the computational steps. However, analyzing at such level of details will quickly become infeasible. Our approach is to abstract away the details and place the error probability in certain key points in the computation. We believe this abstraction strikes the right balance between capturing what is essential on the one hand, while on the other hand retaining a level of simplicity in the model that allows researchers to be able to analyze algorithmic ideas.

\section{Modeling inexactness in the context of Boolean functions}
\label{boolean}
The previous section proposes the general model for inexactness in the general settings. In this section we want to examine the model further in the context of Boolean functions, a fundamental component of computer science theory and practice. In addition, for the next parts of the paper let us consider a more general version of Boolean functions $f:\{0,1\}^n\rightarrow \N$, because this version of Boolean functions are remarkably more popular in computing.

\emph{The overarching theme of  this paper is  inspired by the following:
Given a function $f:\{0,1\}^n\rightarrow \N$, an energy budget  $\Energy$, and a transformation function $\EF$, 
what is the optimal way to distribute $\Energy$ to $n$ segments in order to 
minimize the obfuscation of the reader overall as $f$ is computed. }

Consider an $n$-bit binary vector denoted  $\x = (x_1, x_2, \ldots, x_n) \in \{0,1\}^n$. For an index 
$i\in [n]$, where $[n]$ denotes $\{1, 2, \ldots, n\}$, we use $\Impact{x}{i}$ to denote the vector  that is identical to $\x$ apart from 
the bit $i$, which is ``flipped'' to $1-x_i$. Similarly,
$\x^{(i\rightarrowtail 0)}$ and $\x^{(i\rightarrowtail 1)}$ denote the vectors identical to
 $\x$ with changing only $x_i$ to either $0$ or $1$ respectively, and 
$\x\sim\{0,1\}^n$ denotes a random value $\x$ \emph{uniformly} drawn from 
$\{0,1\}^n$ and so $\x \sim \{0,1\}^n$.
The key concept of  {\em influence} of the $i$th bit for a function $f:\{0,1\}^n\rightarrow \N$  is \[\Inf(i)\triangleq |(f(\x)-f(\x^{\oplus i}))|,\] where $\x$ is drawn uniformly from $\{0,1\}^n$. 
For convenience, we will refer to $\Inf(i)$ and the \emph{influence} of index $i$ without explicitly 
referring to $\x$ when there is no ambiguity. We note that here, we differ  from the
traditional definition of influence~\cite{AnalysisBook2014} which is the expectation of $f:\{0,1\}^n\rightarrow \N$  \emph{with respect to} $\x$
over all uniformly drawn vectors $\x$. However, it is technically more convenient in our case to explicitly express this expectation as $\E [\Inf(i)]$ 
averaged over all uniformly drawn vectors $\x$ and we will adopt this convention in the sequel. 
Furthermore, for convenience, we arrange the input bits so that $\E[\Inf(i)] \le \E[\Inf(i+1)]$ for all $i < n$.

To understand the value of inexactness, influence and its expectation gives us the potential impact of assigning energy to a certain index $i$
preferentially over another index $j$ on the \emph{quality} of the answer. Informally, we wish to assign more energy to variables associated with
indices that have greater expected influence. To formalize this idea, let us define 
the \emph{total impact} of a function  $f$ given an energy vector $\vec{e}=(e_i)_{i\le n}$ and  with induced error probabilities $\vec{p} = (p_i)_{i\le n} = (2^{-e_i})_{i\le n}$ to be 
\[\TIm_f(\vec{p})=\sum_{i\leq n} \E [\Inf(i)]\cdot p_i\] 
We can then use total impact as the measure of how far from the correct values the function drifts given a particular 
energy vector $\vec{e}$. 
Now, given a function $f$ and an energy budget $\Energy$, our goal is to find $\vec{e}=(e_i)_{i\le n}$ that gives the best quality and thus \emph{minimize} 
$\TIm_f$. 

The most obvious and  naive  approach is to consider an energy vector that is \emph{influence oblivious}  
 where we allocate the energy equally to all the indices and therefore,  
$p_i=2^{-\Energy/n}$ for every $i$; this corresponds to the traditional architectural design that treats all bits equally.  In this case,  the \textit{expected total influence oblivious impact}, is 
 \begin{equation}
 \TIm_f(\vec{p}_O) =  2^{-\Energy/n}\sum_{i\leq n}\E[\Inf(i)]
\end{equation}
In contrast, an \emph{influence aware} allocation would 
be guided by the influence values to where indices with higher influence are assigned ``proportionately'' higher energy. Let $\TIm_f(\vec{p}_A)$ be the
\textit{expected total influence aware impact}. Then, to understand the value of inexactness in the context of a function $f$, we define 
a figure of merit  
\begin{equation}
    \alpha = \frac{\TIm_f(\vec{p}_{O})}{\TIm_f(\vec{p}_{A})}
\end{equation} 
the ratio of the total
impact of the oblivious assignment (numerator) to the aware assignment (denominator).  Intuitively, the closer  $\alpha$  is to 1, the less profitable it is to be influence aware as the 
naive  influence oblivious solution can suffice almost as well. Conversely, $\alpha$  being 
large is a strong indication that influence aware solutions are likely to have a much higher 
impact on the  quality of the solution.


To understand this point, let us consider a simple example of evaluating a binary string. Due to binary representation, the impact of an error grows as we progress from the least significant bit to the most significant bit. Thus, we should expect an influence oblivious approach
to perform poorly when compared to one which is influence aware. 
To capture this notion of increased ``weight''
ubiquitous to computer science due to binary numbers we will compare influences as we step through the indices and 
 define 
\begin{equation}\label{eq:betai}
 \beta_i \triangleq \frac{\E[\Inf(i+1)]}{\E[\Inf(i)]}. 
\end{equation}
where $\beta_i$ is the relative influence of index $i+1$ compared to $i$. A straightforward observation is to note that
functions with $\beta_i = 1$ for all $1\le i \le n$ are functions where all the indices are equally influential; we will refer to such functions as being
\emph{influence symmetric}; classical problems such as \emph{parity} and the \emph{OR} function are examples.
 In contrast, \emph{influence asymmetric} functions have $\beta_i > 1$ for some indices $i$. 
 We are particularly interested in functions where all $\beta_i$ values are equal and denoted $\beta$.

\section{Existence of an Optimal Energy Assignment for any Boolean Function}
\label{optimal-existence}

To capture the possible benefits of inexactness aware approaches precisely, we formulate the following problem since the model is new
and we wish to characterize its properties. 

\begin{problem}
\label{prob:inexact}
[Basic Inexactness Problem] We define our {\em basic inexactness problem} comprising a
basic inexact problem instance  and an optimization criterion.
	The  {\em basic inexactness problem instance} is a tuple $(f, \Energy, \EF)$ where 
	$f:\{0,1\}^n\to \N$ is a 
	Boolean function, 
	$\Energy$ is the inexactness energy amount, and 
	$\EF$ is the energy  translation function.
	Given the inexactness problem  instance $(f, \Energy, \EF)$, our {\em optimization criterion} 
	is to find an energy vector 
	$(e_1,\cdots, e_n)$ whose elements sum to at most $\Energy$ (i.e., $\sum_i{e_i}\le \Energy$) 
    such that if 
	$\vec{p}=(p_1,\cdots, p_n)$ where $p_i = \EF(e_i)$ for every $i$, then $\TIm_f(\vec{p})$ 
	is minimized.
\end{problem}

Without loss in generality, we assume that $\EF(e) = 1/2^e$.  We will now show that an 
optimal solution always exists.
\begin{theorem}\label{thm:optim}
For	every inexactness problem with any given any $\Energy$,  a solution that minimizes
the total impact exists and can be computed. 
\end{theorem}
\begin{proof}
Since $p_i=1/2^{e_i}$, we have that $e_i=-\log_2(p_i)$.
Therefore for the constraint $\sum_{i\leq n}e_i = \Energy$ we have:
\begin{equation}
 \sum_{i\leq n}e_i = \Energy \iff \sum(\log(p_i)) = -\Energy \iff  \log(\prod p_i) = -\Energy \iff  
 \prod_{i\leq n} p_i = 1/{2^\Energy}
\end{equation}

Therefore we can redefine the inexact problem in Definition \ref{prob:inexact} as follows

\begin{problem}\label{prob:GP}
	The  problem denoted by $GP(f,\Energy,\EF)$ is to find $\vec{e}=(e_1,\cdots,e_n)$ such that
	\begin{itemize}
		\item $\sum_{i\leq n} \E[\Inf(i)]p_i$ is minimized
		\item $p_i=1/2^{e_i}$
	    \item $\prod_{i\leq n} p_i = 1/{2^\Energy}$
		\item $0 < p_i\leq 1$ for every $i\leq n$ ($p_i$ cannot be $0$)
	\end{itemize}
	
\end{problem}

To solve Problem \ref{prob:GP} we use the AM-GM inequality according to which for every non-negative reals $a_1,\cdots, a_n$ we have.
\begin{equation}
\frac{1}{n}\sum_{1\leq n} a_i \geq (\prod_{i\leq n} a_i)^{\frac{1}{n}}
\end{equation}
Since all the $p_i$ and the $\E[\Inf(i)]$ are non-negative, we can apply the AM-GM inequality to get:
\begin{equation}
\frac{1}{n}\sum_{i\leq n} \E[\Inf(i)]p_i \geq (\prod_{i\leq n} \E[\Inf(i)]p_i)^{\frac{1}{n}}
\end{equation}
Thus,
\begin{equation}
\sum_{i\leq n} \E[\Inf(i)]p_i \geq n(\prod_{i\leq n} \E[\Inf(i)]\prod_{i\leq n}p_i)^{\frac{1}{n}}
\end{equation}
 Since we have a constraint that  $\prod_{i\leq n} p_i = 1/{2^\Energy}$, we all in all have that:
\begin{equation}
\sum_{i\leq n} \E[\Inf(i)]p_i \geq n(\prod_{i\leq n} \E[\Inf(i)] 2^{-\Energy})^{\frac{1}{n}}
\end{equation}

Recall that we need to find values for $p_i$ such that $\sum_{i\leq n} \E[\Inf(i)]p_i$ is minimized.
Since no matter what values of $p_i$ we choose, the left side of the equation above will always be at least as the right side of the equation, minimization will come only when both sides are equal.
In AM-GM we have that equality is reached when $\E[\Inf(i)]p_i$ is identical for all $i$.
Using this we can now establish $p_i$ as follows.
Assuming $\E[\Inf(i)]p_i=k$ for every $i$, we get
\begin{equation}
nk = n(\prod_{i\leq n} \E[\Inf(i)]2^{-\Energy})^{\frac{1}{n}}.
\end{equation}
Thus,
\begin{equation}
k = (\prod_{i\leq n} \E[\Inf(i)] 2^{-\Energy})^{\frac{1}{n}}.
\end{equation}
Therefore for every $i\leq n$ we have 
\begin{equation}
p_i= (\prod_{i\leq n} \E[\Inf(i)] 2^{-\Energy})^{\frac{1}{n}}/ \E[\Inf(i)]
\end{equation}
and setting $e_i=-\log(p_i)$ solves the problem as required.
\end{proof}
It is  not always clear how such an influence aware
energy assignment can be efficiently computed. Even the task of determining whether $\E[\Inf(i)]>0$ 
for a bit $i$ is a co-NP-hard problem as it encompasses asking whether the given CNF 
formula has no satisfying assignment.

\section{Where does inexactness help?}
\label{aware-oblivious}
We have seen that for a Boolean function, there always exists an optimal energy vector that minimizes the total impact. We now  ask when exactly it pays to be influence aware? We shed some light into this question by examining the ratio $\alpha$ of the two extreme cases where for all $i$ $\beta_i = \beta$ a constant greater than 1, and the case where $\beta$ is 1. 
Recall that  $\alpha$ is  the ratio between the expected total influence oblivious impact and the expected total influence aware impact and $\beta_i$ is the ratio between the expected influence of bit $i+1$ and bit $i$.

For a given inexactness problem $(f, \Energy, \EF)$, let 
$\vec{e}=(e_1, e_2, \ldots, e_n)$
be the optimal energy vector (with corresponding error probability vector 
$\vec{p} = (p_1, p_2, \ldots, p_n)$)
 obtained by the influence aware solution.  

\paragraph{Influence Aware Investments}\label{sec:influence}

We now  focus on the important case when all the $\beta_i$ values equal a common constant value
$\beta$. This  special case is in fact quite common and is exemplified by our previous example of evaluating a 
binary bit string where the influence of bit values decrease exponentially, in the context of binary numbers. 


\begin{theorem} \label{thm:alphabeta}
Let $f:\{0,1\}^n \to \mathbb{N}$ be a Boolean function with parameter $\beta > 1$. 
Then, the corresponding $\alpha$ is at least $\Omega\left(\frac{\beta^{n/2}}{n}\right)$, thereby implying
that an influence aware investment is exponentially better (with respect to $n$) than its influence
oblivious counterpart.
\end{theorem}

\begin{proof}
We first have that
$$GM = (\prod_{i<n} \E[\Inf(i)])^{1/n} = (\E[\Inf(1)]^n \prod_{i=0}^{n-1}\beta^i)^{1/n} = \E[\Inf(1)]\beta^{(n-1)/2}$$
and

$$AM = 1/n(\sum_{i\leq n}\E[\Inf(i)]) = 1/n(\E[\Inf(1)]\sum_{i=0}^{n-1}\beta^i) = \frac{\E[\Inf(1)]}{n} \frac{\beta^n-1}{\beta-1}$$

Therefore we have that

$$\alpha = AM/GM = \frac{\beta^n-1}{n\beta^{(n-1)/2}(\beta-1)}$$

This ratio is $\Omega(\frac{\beta^{n/2}}{n})$ when $\beta > 1$.
\end{proof}

To continue the example of evaluating $n$-bit binary strings, we present the following.  

\begin{corollary}\label{exm:BE}
The value of $\alpha$ for \textit{Binary Evaluation (BE) function} $f:\{0,1\}^n\to \mathbb{N}$ that 
takes a binary input and returns the decimal evaluation of that input is at least $\Omega\left(\frac{2^{n/2}}{n}\right)$.
\end{corollary}

\paragraph{Influence Oblivious Investments}\label{sec:infFunctions}
 
To reiterate, formally, a Boolean function $f$ is  influence-symmetric if all of the bits 
of $f$ have the same influence (i.e.,  $\Inf(i)=\Inf(j)$ for every $i \ne j$). 
Recall that from the definition of  $\alpha$,   we see that if $f$ is an influence-symmetric function 
then $\alpha=1$. An important class of  Boolean functions are the \textit{symmetric Boolean functions} 
defined as the set of functions $f$ such that for all $\x$ and any permutation $\sigma$, 
$f(\x) = f(\sigma(\x))$ and therefore  changing the order of the bits does not 
change the output of the function. We now have:  

\begin{theorem}\label{cor:sym}
The influence oblivious assignment is an optimal energy distribution for 
influence-symmetric Boolean functions. Furthermore, every symmetric function is also influence-symmetric,
so an influence oblivious investment is optimal.
\end{theorem}

\begin{proof}
The key observation that we need to make first is that symmetric functions $f$ can be evaluated
just by counting the number of 1's in the input. Let $d_j \triangleq f(x^j)$, where $x^j\in \{0,1\}^n
$ has exactly $j$ 1's. Let us consider $\Inf(1)$. What is the probability that $\Inf(1)$ takes the 
value, say, $a$? To answer this, let $J_a \subset [n]$ denote the set of all $j$ such that
$|d_j- d_{j-1}| = a$.  Then, clearly, $\Inf(1) = a$ with probability 
$\sum_{j \in J_a} \frac{\binom{n-1}{j-1}}{2^n}$. This same argument can be repeated for 
obtaining  $\Pr[\Inf(i)=a]$. Thus, the random variables $\Inf(i)$ and $\Inf(i')$, $i \ne i'$, 
have the same distributions, thereby implying that $f$ is influence-symmetric. 
\end{proof}

Let us consider the parity function $\xor(\x)$ -- a quintessential symmetric function -- 
that outputs 1 when the number of 1's in  $\x$ is odd, and $0$ otherwise. In this case,
the $\Inf(i) = 1$ for all $i$.  Thus, the $\TIm(f, \Energy) = n2^{\Energy/n}$, and this 
matches the total influence of the assignment that is influence oblivious wherein each
bit is assigned energy $\Energy/n$. Thus, $\alpha = 1$ for $\xor$.

\section{The influence ratio and PAC learning}\label{sec:learning}
Machine learning has been one of the most popular topics in computer science for decades. In this section, we would like to establish a direct relation between the influence ratio $\beta$ and a widely studied form of theoretical machine learning called Probably Approximate Correct (PAC) Learning where Boolean functions are learned with some margin for error. By exploring the relation 
between the concepts of fixed $\beta$ and PAC learning, we show that a function is 
more PAC-learnable if its influence ratio is greater than 1. \emph{Thus, this establishes the connection between the cases where machine learning performs well and the cases which can benefit from an influence aware approach}.

In this section, we use an alternative more general form of Boolean functions $f: \{-1, +1\}^n \to \mathbb{R}$. This form is widely used in the study of PAC learning and analysis of Boolean functions. For a subset $S\subseteq[n]$ let $x^S=\prod_{i\in S}x_i$ where every $x_i\in\{-1,1\}$. 

For a Boolean function $f$, another way to describe $f$ is as a mutli-polynomial called the Fourier expansion of $f$,
$$
f(\x)=\sum_{S\subseteq[n]}\hat{f}(S)x^S
$$
where the real number $\hat{f}(S)$ is called the coefficient of $f$ on $S$. Then we have from~\cite{AnalysisBook2014} that the influence of a bit $i$ is as follows.
\begin{definition}
 Define $Var(i)= \frac{1}{4}E[\Inf(i)^2]$,
\end{definition}
\begin{claim}\label{clm:prem}
For every $i\leq n$ $Var(i) =\sum_{S\subseteq[n], i\in S}\hat{f}(S)^2$.
\end{claim}

 The proof, as well as more on analysis of Boolean functions can be found in~\cite{AnalysisBook2014}.

\begin{definition}
	Let $\epsilon>0, 0\leq k\leq n$. A function $f:\{0,1\}^n\rightarrow \Re$ is called $\epsilon$-concentrated up to degree $k$, if $\sum_{S\subseteq [n], |S|>k} \hat{f(S)}^2 <\epsilon$.
\end{definition}


The notion of $\epsilon$-concentration up to degree $k$ is particularly interesting as it allows us to efficiently learn the function.

\begin{theorem} [The \textbf{"Low-Degree Algorithm"} from page 81 in \cite{AnalysisBook2014}] Let $k \ge 1$ and let $C$ be a concept class for which every function $f : \{-1,1\}^n \rightarrow \{-1,1\}$ in $C$ is $\epsilon/2$-concentrated up to degree $k$. Then $C$ can be learned from random examples only with error $\epsilon$ in time
poly$(n^k,1/\epsilon)$.
    
\end{theorem}

We are ready to state our main theorem, which states that having influence ratio $\beta > 1$ implies PAC learnability. 
 
Note that quite often we do not get an exact measure $\beta$, for the influence ratio $1 \le \beta \le  \frac{ \Inf(i+1)}{\Inf(i)}$ of a function, and it is much easier to obtain lower and upper bounds $\beta_1,\beta_2$ such that $1 \le \beta_1 \le \beta \le \beta_2$. 
By using these assumptions, and an additional bound $\Inf$ for which $\Inf(n)<\Inf$ we can say the following \footnote{The bounds $\beta_1,\beta_2,\Inf$ can be learned by means such as random sampling.}. Given $\epsilon>0$ calculate constant
$k>0$ such that

\begin{equation}\label{eq:k-learn}
	\Inf\cdot\frac{\beta_1^{-k}}{\beta_1-1} <\epsilon/2
\end{equation}
For simplicity denote $\Inf^f(n)$ the parameter $\Inf(n)$ for the specific function $f$. Then we have.

\begin{theorem} \label{thm:learning}
	Let $C$ be a concept class and  $\beta_1, \beta_2 \geq 1$, such that every function  $f:\{-1,1\}^n\to\{-1,1\}$ in $C$ has $\beta_f$ where $\beta_1\leq\beta_f\leq \beta_2$, and such that $\Inf^f(n)\leq \Inf$ for some given parameter $\Inf>0$,  and let $\epsilon>0$. Then $C$ can be learned from random examples with only error $\epsilon$ in time $poly(n^k,1/\epsilon)$.
\end{theorem}
The proof of the theorem can be found in Appendix \ref{appendix:learning}.

\section{Modeling inexactness in the context of sorting} \label{sec: sort}
We have been studying the idea of inexactness applied to Boolean functions. However, in real world applications, not all computational tasks are Boolean functions. Hence in this section, we examine the problem of sorting, an important computing task, to illustrate the benefit of influence aware investments.

Here we employ a setting wherein the data is an array $C$ of $N$ items  stored in the ``cloud'' and a local computer called the client must compute a sorted ordering of the data.  We begin with each data item $C[j]$, $1 \le j \le N$, being  $n$ bits   drawn uniformly at random from $\{0,1\}^n$ representing integers in the range $[0, 2^n-1]$. Since $C$ is in the cloud, the client can only access the data items indirectly through a predefined functions  $\com{a}{b}{\vec{e}}$, where $a$ and $b$ are two indices of  the array $C$ and $\vec{e}$ is an energy vector. Since comparison seeks to find the most significant bit in which $C[a]$ and $C[b]$ differ, it employs bit-wise comparison. Thus, $\vec{e}$ serves the purpose of apportioning energy values across the bits.  The client's goal is to compute a permutation of $[N]\triangleq \{1, 2, \ldots, N\}$ that matches the sorted ordering of $C$.

Our outcome will be an approximation of the correctly sorted ordering where, in the spirit of inexactness, minor errors that wrongly order numbers close in magnitude are more acceptable than egregious errors that reorder numbers that differ a lot. Thus, we measure sortedness using a measure that we call the {\em weighted Kendall's $\tau$ distance}~\cite{GP18} that we now seek to define. We establish some notations first.
Let $C^*$ denote the sorted permutation of the arbitrary array $C$. Consider two indices $a$ and $b$, both in the range $[1,N]$. Let $X(a, b)$ be an indicator random variable that is 1 when $C[a]$ and $C[b]$ are ordered differently in $C$ and $C^*$, and 0 when they are ordered the same way. The classical Kendall's $\tau$ distance~\cite{K38} counts the number of inversions and is defined as $\sum_{a \ne b} X(a, b)$. We are however interested in the weighted Kendall's $\tau$ distance of $\pi$ denoted $\wkt{C}()$ and it is defined as 
\begin{equation} \label{eq:wkt1}
    \wkt{C} \triangleq \sum_{a \ne b} \left [|C[a] - C[b]| \cdot X(a, b) \right ]
\end{equation}
The intuition behind this measure is that bigger difference between two numbers having incorrect relative order should result in bigger penalty, and vice versa. A reasonable inexact comparison scheme should have a smaller error chance for numbers that are farther apart - this is correct in the case of comparing using inexactness aware energy allocation scheme as we will see in this section.

Now, let us abuse the notation and use $\wkt{C}(\e)$ to denote the expected weighted Kendall's $\tau$ distance of the permutation that we receive when we perform quicksort on input array $C$ using energy vector $\e$. Note that this value is averaged over all the runs of quicksort, with the random factors being the pivot choices of quicksort and the comparison error from inexactness:
\begin{equation} \label{eq:wkt2}
    \wkt{C}(\e) \triangleq \E_{C^\e}\left[\wkt{C^\e}\right]
\end{equation}
where $C^\e$ denotes a permutation of array $C$ after quicksort using energy vector $\e$.

In the next part of this section we are interested in analyzing  the ratio of the expected weighted Kendall's $\tau$ distance using inexactness oblivious energy to its inexactness aware energy counterpart (the expectation is taken over all possible input arrays $C$)
\begin{equation}
\alpha^* \triangleq  \E_{C}[\wkt{C}(\eo)]/\E_{C}[\wkt{C}(\ea)]\,    
\end{equation}
which is analogous to the ratio $\alpha$ defined for Boolean functions. Our goal is to show that this ratio grows exponentially (in $n$). For the energy aware case the energy vector $\ea = (1, 2, \ldots, n)$, thereby assigning higher energy values to higher order bits. 
On the other hand, for the energy oblivious case, we use equal energy for all the bits, so the energy vector is $\eo = (\frac{n+1}{2}, \frac{n+1}{2}, \cdots, \frac{n+1}{2})$. Both the inexactness aware and the inexactness oblivious algorithms employ quicksort using $\com{\cdot}{\cdot}{\cdot}$ functions, but with their respective energy vectors. 

In the sequel theorems and proofs, we will use $I(a,b,\e)$ to denote the event that the comparison between two numbers $a$ and $b$ is incorrect using the energy vector $\e$. We use $Q(a,b,\e)$ to denote the event that the quicksort algorithm with input $C$ using energy vector $\e$ results in two numbers $a$ and $b$ having the incorrect relative positions (for simplicity we omit the input array $C$ from this notation). For simplicity, we will assume that the elements in our input array are distinct. Finally, from its definition in equations \ref{eq:wkt1} and \ref{eq:wkt2} and from linearity of expectation, $\wkt{C}(\e)$ can be calculated as follows
\begin{equation} \label{eq:wkt3}
    \wkt{C}(\e) =  \sum_{1 \le a < b \le n} |C[a] - C[b]| \cdot \Pr[Q(C[a], C[b], \e)]
\end{equation}

We state our desired result of the lower bound of $\alpha*$ as follows. The proof of this theorem can be found in Appendix \ref{appendix:sorting1}.
\begin{theorem} \label{thm: sorting1}
The ratio $\alpha^*$ is $\Omega(\frac{2^{n/2}}{N \log N})$.
\end{theorem}


The advantage of influence-aware approach can be shown not only through the ratio $\alpha*$ which is based on the difference between the two approaches' average weighted Kendall's $\tau$ distance over all inputs, but also through the distance difference of the majority of individual inputs. Let us define a good input as one for which the inexactness aware assignments results in an exponentially lower value of weighted Kendall's $\tau$ distance; the rest of the inputs are called bad. More specifically, a good input $C$ is such that the ratio between $\wkt{C}(\eo)$ and $\wkt{C}(\ea)$, the expected weighted Kendall's $\tau$ of quicksort with input $C$ under energy oblivious and energy awareness, is $\Omega(\frac{2^{n/6}}{N\log N})$. Let $g$ and $b$ denote the number of good and bad inputs, respectively. We will show that $g/b$ is exponential in $n$.

\begin{theorem} \label{thm: sorting_input}
The ratio of the number of good vs. bad inputs is  $\Omega(\frac{2^{n/3}}{N^2})$. Therefore, as $n \to \infty$, $g/b \to \infty$.
\end{theorem}
The proof of this theorem can be found in Appendix \ref{appendix:sorting2}

\section{Variable Precision Computation}
\label{sec: truncation}
In practice, manufacturers usually lack the resource to assign a different level of energy to every bit in a chip. A more practical approach is that only $\gamma$ different levels of energy are assigned to the bits, usually with the lowest energy level being 0. This approach is usually referred to as \textit{variable precision computation}, and has been studied in some works such as ~\cite{Hittinger2019VariablePC}, owing to its simplicity and effectiveness.

In this section, we will focus in the scenario where $\gamma = 2$, i.e. a large proportion of the energy is equally focused on the most significant $\frac{n}{k}$ bits, where a small proportion, if not none, of the energy is assigned to the remaining $n(1 - 1/k)$ bits. We denote this energy vector $\et$. Since the total energy is $\approx n^2/2$, $\et = \{ 0, 0, \ldots 0, 2nk, \ldots 2nk \}$. 

The goal of this section is to study the effect of using this energy vector compared to the inexactness oblivious approach for basic functionalities, which we again use sorting and the weighted Kendall's $\tau$ metric as an example. We are interested in bounding the value of the ratio between ${\E_C[\wkt{C}(\eo)]}$ $/{\E_C[\wkt{C}(\et)]}$ which is analogous to $\alpha^*$ in Section \ref{sec: sort}, and the ratio between good and bad inputs, whereas bad (and good) inputs are generally the ones that make $\wkt{C}(\eo)$ $/{\wkt{C}(\et)}$ exponential in $n$ following the convention in Section \ref{sec: sort}. 

Toward that goal, we prove the following two theorems. The combined result of the two theorems gives us an estimate of a 'good' truncation ratio $k$, which is inside the interval $(\frac{5}{3}, 4)$.

\begin{theorem} \label{thm: truncate1}
Let $k$ be a parameter and assume we use an energy allocation scheme $\et$ where energy is divided equally on $\frac{n}{k}$ most significant bits. Then, for an arbitrary input array $C$ drawn from the uniform random distribution, $$\Pr\left[\frac{\wkt{C}(\eo)}{\wkt{C}(\et)} = O(\frac{2^{n(k - 5/3)/6}}{N\log N})\right] = O(\frac{N^2}{2^{\frac{n}{max(3,k)}}})$$
Consequently, for constant $k > 5/3$, if we define bad inputs to be the ones that make the ratio $\frac{\wkt{C}(\eo)}{\wkt{C}(\et)}$ $O(\frac{2^{n(k - 5/3)/6}}{N\log N})$ and good inputs to be the remaining, then the ratio between good and bad inputs is at least $\Omega(\frac{2^{n/max(3,k)}}{N^2})$.
\end{theorem}

\begin{theorem} \label{thm: truncate2}
Let $k$ be a parameter and assume we divide energy on $\frac{n}{k}$ most significant bits. Then, for $k < 4$, the ratio $\frac{\E_C[\wkt{C}(\eo)]}{\E_C[\wkt{C}(\et)]}$ is exponential in $n$.
\end{theorem}

\textbf{Remarks} The variable precision energy allocation scheme is a more practical approach to inexactness where the energy is focused only on the most significant $n/k$ bits. In this section, we have shown that for a value of $k$ in the interval $(\frac{5}{3}, 4)$, sorting using variable precision energy allocation scheme is exponentially better than using inexactness oblivious energy allocation in the weighted Kendall's $\tau$ metric. This is true for both the average case (Theorem \ref{thm: truncate2}) and for most of the possible inputs with only an exponentially small number of exceptions\footnote{Of course, input data will often depend on the particular application at hand and may not be immediately suitable for variable precision in the manner we have presented. However, we believe that the principle can be adapted to work nevertheless.} (Theorem \ref{thm: truncate1}).  Note that the specific value $\frac{5}{3} < k < 4$ resulted from the analysis with $\approx \frac{n^2}{2}$ total energy. For the analyses using different levels of total energy and different restrictions (such as the number $\gamma$ of distinct energy levels), we might arrive at different schemes of energy distribution. Nevertheless, the core principle of inexactness should remain applicable.

\section{Concluding remarks}\label{sec:discussion}


The algorithmic end of computing has a rich history of  examples such as  \emph{randomization}~\cite{rabin, solovay-strassen}  
and approximation algorithms~\cite{vazirani-book}, and combined approaches such as \emph{fully polynomial Randomized approximation schemes (FPRAS)}~\cite{KL83}
 which departed radically from traditional computing philosophy of guaranteeing correctness. 
Specifically, they embraced the possibility 
that computations can yield results that are not entirely correct while offering (potentially) significant savings in resources consumed, typically running time.
Despite this relaxed expectation on the the quality of their solutions, 
randomized and approximation algorithms were always deployed on reliable 
computing systems. In contrast, inexact computing crucially differs by advocating the use of ``unreliable'' 
computing architectures and systems directly and thus, blend in the behavior 
from the platform on which it is executing directly into the algorithm. 
Thus, one can view the inexactness in our model as a way of extending the principles of randomization and approximation down to the hardware level, thereby improving the overall gains that we can garner.
Thus, the ability to lower cost by lowering energy,  and
its allocation to different parts of the computation guided by influence are made explicit and  can be managed by the algorithm designer. 
By demonstrating the value of this idea in canonical and illustrative settings, namely theory of Boolean functions, PAC learning, inexact sorting, we aimed to have
demonstrated its value in a range of settings. In principle, the model we have introduced and whose value we demonstrated through several foundational building blocks is truly general in the following sense: \emph{given any computing engine and hence an instance of our model, an algorithm can be designed and evaluated. Additionally, due to its theoretical generality, our model parameters allow us to assert the cost and quality of algorithms as functions of parameter values and thus can, in the spirit of the foundations of computer science, be characterized as theorems are true asymptotically.}

In addition to extending the notion of randomized and approximate computation to the hardware level, we believe that the framework of inexactness that we have introduced can seamlessly extend beyond its immediate motivation from CMOS technology. At its core, the potential for inexactness stems from the notion of influence that is  orthogonal to the computing technology that is employed. In our work, we have framed the model using CMOS principles and the concomitant error function that decays exponentially with energy. Alternative technologies like quantum computing may offer slightly different modeling parameters, but we believe that the core principles based on the notion of influence will remain intact and effective.
Thus, we hope that our work will enable future  work resulting in a principled injection of inexactness in a wide range of contexts.



\bibliographystyle{plain}
\bibliography{inexact}
\newpage

\appendix
\section{Proof of theorem \ref{thm:learning} in section \ref{sec:learning}}
\label{appendix:learning}

We show that given $(\epsilon, f)$, by obtaining the fixed $\beta$ from the function $f$ we can provide $k$ such that $f$ is concentrated with respect to $(\epsilon, k)$.
The close relation between a concentrated and learnable functions will lead us to a result about the capability of the function of being learnable.

\begin{lemma}
	Let $f$ be a function $f:\{0,1\}^n\rightarrow \Re$	such that $\sum_{i>k}Var(n - i)<\epsilon$ for $\epsilon>0$ and $k\leq n$.
	Then $f$ is $\epsilon$-concentrated up to degree $k$.
\end{lemma}

\begin{proof}
	From Claim \ref{clm:prem}, for every $i<n$ we have
	
	\begin{equation}
		Var(n-i)=\sum_{S\subseteq [n], n-i\in S} \hat{f(S)}^2
	\end{equation}

	Now let $S$ be a subset of size bigger than $k$. Then there must be an element $i>k$ such that $n-i\in S$.
	Therefore $\{S| S\subseteq [n], |S|>k\}\subseteq \{S| S\subseteq [n], \exists(i>k) n-i\in S\}$.
	Therefore we have that
	\begin{equation}
		\sum_{S\subseteq [n], |S|>k} \hat{f(S)}^2 \leq \sum_{S\subseteq [n], \exists(i>k) n-i\in S} \hat{f(S)}^2 = \sum_{i>k}\sum_{S\subseteq [n], n-i\in S} \hat{f(S)}^2= \sum_{i>k}Var(n-i) <\epsilon
	\end{equation}
	
	Therefore $f$ is $\epsilon$-concentrated up to degree $k$. 
\end{proof}


Given that the influences $\Inf(i)$ are all positive random variables we have an easy lemma.
\begin{lemma}
\label{lem:cute}
    	Let $f$ be a function $f:\{0,1\}^n\rightarrow \Re$	such that $\sum_{i>k}\Inf(n-i)<\epsilon$ for $1 > \epsilon>0$ and $k\leq n$.
	Then $f$ is $\epsilon$-concentrated up to degree $k$.
\end{lemma}
\begin{proof}
    $\sum_{i>k}\Inf(n-i)<\epsilon < 1$ implies $\sum_{i>k}Var(n-i)<\epsilon$ because of $ 0 \le \Inf(i) < 1$ from the definition of influence.
\end{proof}

With Lemma~\ref{lem:cute}, we are ready to prove our main theorem of this section.
\begin{proof}[Proof of Theorem \ref{thm:learning}]
For every function $f$ in $C$ we have that $\beta_1 \le \beta_f$. Then we have: 
	
	\begin{equation}\label{eq:k-learnTmp}
		\sum_{i>k}\Inf(n-i) =\Inf(n)\frac{\beta_f^{-k}-\beta_f^{-n+1}}{\beta_f - 1} < \Inf\frac{\beta_1^{-k}}{\beta_1-1} <\epsilon/2
	\end{equation}

	Therefore we have from Lemma \ref{lem:cute} that every function $f$ in $C$ is  $\epsilon/2$-concentrated up to degree $k$. By the ``Low-Degree'' Algorithm, the learnability from random examples with error $\epsilon$ in time $poly(n^k,1/\epsilon)$ follows immediately. 
    
\end{proof}

It is interesting to note that the opposite of Lemma \ref{lem:cute} is not always correct. Specifically,  it is easy to construct a function $f$ for which $\hat{f(\{i\})}=1$ for every bit $i$, and $\hat{f(S)}=0$ (or is very small) for every set $S$ of size at least $2$.
Such a function is $\epsilon$-concentrated up to degree $2$ but no matter what the order of bits are for every $i$, $\Inf(i) \geq \hat{f(\{i\})}=1$ So for every $k<n$ $\sum_{n-i>k}\Inf(n-i) \geq 1$.
Studying of when exactly a small concentration lead to small influence ratio is a matter of future work.

\section{Proof of Theorem \ref{thm: sorting1} in section \ref{sec: sort}}
\label{appendix:sorting1}
First we prove some lemmas which will be useful for the rest of this section.

\begin{lemma} \label{lemma: sorting 1.1}
Consider two arbitrary $n$-bit integers $a < b$. Then
$$\Pr[I(a,b, \ea)] < \frac{8}{b-a}$$
\end{lemma}
\begin{proof}
Consider $n \ge i \ge 1$ the first bit where $a$ and $b$ differ ($i = n$ means $a$ and $b$ differ from the most significant bit).

The probability that comparison is wrong before getting to bit $i$ is less than or equal to the sum of the probabilities that the comparison is wrong at bit $j > i$, which is $= 1/2^{n} + 1/2^{n-1} + \ldots + 1/2^{i+1} = 1/2^{i} - 1/2^{n}$.

The probability that the comparison is wrong at bit $i$ is $1/2^{2i}$ since the bits $i$ in $a$ and $b$ must both be flipped for this to happen. 
The probability that the comparison is wrong after bit $i$ is $< 1/2^{i-1}$ as exactly one bit $i$ in either $a$ and $b$ must be flipped and the remaining of $C[a]$ has a probability $1/2$ to be compared bigger than the remaining of $b$. 

Summing all these, we have $\Pr[I(a,b)] < 1/2^{i-2}$.

Note that when $i$ is the first bit that $a$ and $b$ differ, we have $2^{i+1} > b - a \ge 1$. 
Therefore, $(b-a)\Pr[I(a,b)] < \frac{1}{2^{i-2}}\cdot 2^{i+1} = 8$ and the result follows.

\end{proof}

\begin{lemma} \label{lemma: sorting 1.3}
Consider an arbitrary array $C$ of $N$ $n$-bit integers. Then
    $$\wkt{C}(\ea) =  O(N^2 \log N)$$
\end{lemma}

\begin{proof}
We will prove the inequality using result from Lemma \ref{lemma: sorting 1.1}. Consider a random pivot $C[j]$. Note that once the pivot $C[j]$ splits $C[a]$ and $C[b]$ to two sides, then the relative position of $C[a]$ and $C[b]$ is determined. In particular if $C[j]$ splits $C[a]$ to the bigger side and $C[b]$ to the smaller side then the event that split is incorrect at pivot $C[j]$, which we denote $Q(C[a], C[b], C[j],\ea)$, happens. There are 4 cases:
\begin{itemize}
    \item Case 1: $C[a] < C[b]< C[j]$. In this case $Q(C[a],C[b], C[j], \ea)$ happens if only $I(C[a], C[j])$ happens. Therefore in this case 
    \begin{equation*}
        \begin{split}
            \Pr[Q(C[a],C[b], C[j], \ea)] &= \Pr[I(C[a],C[j], \ea)] \cdot(1-\Pr[I(C[b],C[j], \ea)]) \\
            &< \Pr[I(C[a],C[j], \ea)] < \frac{8}{C[j] - C[a]} < \frac{8}{|C[a] - C[b]|} 
        \end{split}
    \end{equation*} 
    \item Case 2: $C[j] < C[a] < C[b]$. In this case $Q(C[a],C[b],C[j], \ea)$ happens if only $I(C[b],C[j], \ea)$ happens. Therefore again 
    \begin{equation*}
    \begin{split}
        \Pr[Q(C[a],C[b],C[j], \ea)] & = \Pr[I(C[b],C[j], \ea)].(1-\Pr[I(C[a],C[j], \ea)])\\
        &< \Pr[I(C[b],C[j], \ea)] <  \frac{8}{C[j] - C[b]} < \frac{8}{|C[a] - C[b]|}  
    \end{split}
    \end{equation*}
    \item Case 3: $C[a] < C[j] < C[b]$. In this case $Q(C[a],C[b],C[j], \ea)$ happens if both $I(C[b],C[j], \ea)$ and $I(C[a],C[j], \ea)$ happen. Therefore, noting that $\Pr[I(x,y, \ea)] < 1/2$ for all $x \ne y$
    \begin{equation*}
    \begin{split}
    \Pr[Q(C[a],C[b],C[j], \ea)] &= \Pr[I(C[b],C[j], \ea)] \cdot \Pr[I(C[a],C[j], \ea)]\\
    &< \frac{8}{C[j] - C[a]} \cdot\frac{8}{C[b] - C[j]} < \frac{64}{C[a] - C[b]}
    \end{split}
    \end{equation*}
    \item Case 4: either $a = j$ or $b = j$. In this case, $\Pr[Q(C[a],C[b],C[j], \ea)] = \Pr[I(C[a],C[b], \ea)]$.
\end{itemize}
Therefore, for all choices of pivot $C[j]$, $\Pr[Q(C[a], C[b], C[j], \ea] < \frac{64}{|C[a] - C[b]|}$. The probability of $Q(C[a],C[b], \ea)$ happens is less than or equal to the sum of the probabilities that $Q(C[a], C[b], C[j], \ea)$ happens for $C[j]$ all the pivots that $C[a]$ and $C[b]$ get compared to
\begin{align*}
\Pr[Q(C[a],C[b], \ea)] \le \sum_{j}\Pr[Q(C[a],C[b],C[j], \ea)] \\
\le 1.4\log N\Pr[I(C[a],C[b], \ea)] < 1.4 \log N \frac{64}{|C[a] - C[b]|}
\end{align*}
using the well-known result that the expected recursion depth of quicksort is $< 1.4\log N$ (here we assume that the worst case of the version of quicksort that we use is not input-dependent). Therefore, 
\begin{equation}
    \Pr[Q(C[a],C[b], \ea)]\cdot|C[a] - C[b]| < 90 \log N
\end{equation}
for arbitrary $a,b$.

We can now evaluate $\wkt{C}(\ea)$ as follows.
\begin{equation} \label{eq3}
\begin{split}
\wkt{C}(\ea) & = \sum_{1 \le a < b \le N}\left [ |C[a] - C[b]| \cdot  \Pr[Q(C[a],C[b], \ea)]\right ]\\
& \le \sum_{1 \le a < b \le N} 90 \log N\\
& < 50 N^2 \log N.
\end{split}
\end{equation}
\end{proof}
Now we can use the lemmas developed to prove our main result of this section.
\begin{proof}[Proof of Theorem \ref{thm: sorting1}]
From Lemma \ref{lemma: sorting 1.3} we know that for every input array $C$, $\wkt{C}(\ea) < 50 N^2 \log N$. Thus, 
\begin{equation} \label{eq1}
    \E_C[\wkt{C}(\ea)] < 50 N^2 \log N
\end{equation}

We now turn out attention to bounding $\E[\wkt{C}(\eo)]$ from below. Consider any input array $C$, two arbitrary indices $a, b$ and the first pivot of quicksort $C[j]$. The probability that $j = a$ or $j = b$ is $2/N$, and when this happens $\Pr[Q(C[a], C[b], C[j], \eo) = \Pr[I(C[a], C[b], \eo)]$. Thus 
\begin{equation} \label{eq: sorting 1.7}
\Pr[Q(C[a], C[b], \eo)] > \Pr[Q(C[a], C[b], C[j], \eo) > \frac{2}{N}\cdot \Pr[I(C[a], C[b], \eo)] 
\end{equation}
Now, because we are reasoning about expectation over all inputs $C$, we have:
\begin{equation} \label{eq4}
\begin{split}
\E_C[\wkt{C}(\eo)] & = \frac{N(N+1)}{2}\cdot \E_{0 \le c,d \le 2^n-1}[|c - d| \cdot \Pr[Q(c,d,\eo)]] \\
& > \frac{N(N+1)}{2}\cdot \frac{2}{N} \E_{0 \le c,d \le 2^n-1}[|c - d| \cdot \Pr[I(c,d,\eo)]]
\end{split}
\end{equation}

Let us now bound $\E_{0 \le c,d \le 2^n-1}[|c - d| \cdot \Pr[I(c,d,\eo)]]$. As before, consider $1 \le i \le n$ as the index of the leading different bit of $c$ and $d$, denote $FD(c,d) = i$. We will reason about $\Pr[I(c,d, \eo)]$. When $i < n$ the probability that the comparison is wrong is $> 1/2^{n/2 + 2}$ since the comparison can be wrong at the very first bit. Thus, 

\begin{equation} \label{eq4}
\begin{split}
&\E_{0 \le c,d \le 2^n-1}[|c - d| \cdot \Pr[I(c,d,\eo)]] \\
 & = \E_{FD(c,d) = n}|c-d| \cdot \Pr[I(c,d, \eo) | FD(c,d) = n]\cdot \Pr[FD(c,d) = n]  \\
 &+ \E_{FD(c,d) < n}|c-d| \cdot \Pr[I(c,d, \eo) | FD(c,d) < n]\Pr[FD(c,d) < n] \\
 &> \E_{FD(c,d) = n}|c-d| \cdot \Pr[I(c,d, \eo) | FD(c,d) = n]\cdot \Pr[FD(c,d) = n]\\
 &= \frac{2^{n-1}}{3}. \frac{1}{2^{n/2 + 2}}\cdot \frac{1}{2} =  \frac{2^{n/2-4}}{3}
\end{split}
\end{equation}

Therefore, from equation \ref{eq3} and equation \ref{eq4}
\begin{equation} \label{eq5}
\E_C[\wkt{C}(\eo)] >  \frac{N+1}{48}2^{n/2} 
\end{equation}

The theorem follows from equation \ref{eq1} and equation \ref{eq5}.

\end{proof}

\section{Proof of Theorem \ref{thm: sorting_input} in section \ref{sec: sort}}
\label{appendix:sorting2}
First let us prove a lemma
\begin{lemma}\label{lemma: sorting2.1}
Consider an input array $C$ of $N$ integers drawn from the uniform random distribution over the set of all inputs. Then, $$\Pr\left[\frac{\wkt{C}(\eo)}{\wkt{C}(\ea)} = \Omega(\frac{2^{n/6}}{N\log N})\right] = 1 - O(\frac{N^2}{2^{n/3}})$$
\end{lemma}

\begin{proof}
From Lemma \ref{lemma: sorting 1.3} we know that $\wkt{C}(\ea) < 50 N^2 \log N$ for every input $C$. Thus we only need to calculate a concentration bound of $\wkt{C}(\eo)$ to arrive at our goal.

From equation \ref{eq: sorting 1.7} in the proof of Theorem \ref{thm: sorting1}, we know that for every pair of indices $a$ and $b$, $\Pr[Q(C[a], C[b], \eo)] > \frac{2}{N}$ $\cdot \Pr[I(C[a], C[b], \eo)]$, and thus 
\begin{equation} \label{eq2.1}
\begin{split}
    \wkt{C}(\eo) &= \sum_{a < b\le N} |C[a] - C[b]| \cdot \Pr[Q(C[a], C[b], \eo)] \\ 
    &> \frac{2}{N} \sum_{a < b\le N} |C[a] - C[b]| \cdot \Pr[I(C[a], C[b], \eo)]
\end{split}
\end{equation}

Let us consider the case when $C[a]$ and $C[b]$ have the first different bit index $< n$, i.e. $FD(C[a], C[b]) < n$. As stated before in the proof of Theorem \ref{thm: sorting1}, when $FD(C[a],C[b]) < n$ we have $\Pr[I(C[a], C[b], \eo)] > 1/2^{n/2 + 2}$. For $\frac{N(N+1)}{2}$ pairs $C[a], C[b]$ there are at least $N(N+2)/2$ pairs having the same leading bit (this is a straight application of AM-GM inequality). On the other hand, if $C[a]$ and $C[b]$ are random numbers from the uniform distribution with the same leading bit, $|C[a] - C[b]|$ can be seen as the absolute difference of 2 uniformly random $(n-1)$-bit numbers. Now, for two $(n-1)$-bit numbers $c < d$ uniformly random, the probability distribution of $(d-c)$ is a triangle that connects points $(0,0)$, $(0, 2^{n})$ and $(2^{n}, 0)$ in the coordinate system. From this probability distribution, we solve the probability that $d-c\le 2^{2n/3}$ or $b-a \ge 2^{n} - 2^{2n/3}$ is $\frac{1}{2^{n/3}}$.

Thus from equation \ref{eq2.1} and the union bound, with probability $> 1 - \frac{N^2}{2^{n/3}}$ we have
\begin{equation} \label{eq6}
\begin{split}
    \wkt{C}(\eo) 
    &> \frac{2}{N} \sum_{a < b\le N} |C[a] - C[b]| \cdot \Pr[I(C[a], C[b], \eo)] \\
    &> \frac{2}{N} \cdot \frac{N^2}{2} \cdot \frac{1}{2^{n/2 +2}} \cdot 2^{2n/3} \\
    & = \frac{N}{8}\cdot 2^{n/6}
\end{split}
\end{equation}

Thus, from Lemma \ref{lemma: sorting 1.3}
\begin{equation}
    \Pr\left[\frac{\wkt{C}(\eo)}{\wkt{C}(\ea)}  > \frac{N\cdot 2^{n/6}}{400 N^2 \log N}\right] > 1 - \frac{N^2}{2^{n/3}}
\end{equation}
and the result follows.
\end{proof}
Our desired result follows immediately after the above lemma.
\begin{proof}[Proof of theorem \ref{thm: sorting_input}]
From Lemma \ref{lemma: sorting2.1} the probability of an input $C$ drawn from the uniform random distribution over the set of all inputs being bad is $O(\frac{N^2}{2^{n/3}})$ and the theorem follows.
\end{proof}

\section{Proof of Theorem \ref{thm: truncate1} and Theorem \ref{thm: truncate2} in section \ref{sec: truncation}}
Similar to the previous section, we will first prove some useful lemmas.
\begin{lemma} \label{lemma: sorting2.1}
Given two random number $a$ and $b$, if the first different bit index $FD(a,b) = i \ge n - n/k$ then 
\begin{itemize}
    \item $\Pr[I(a,b,\et)] < \frac{n-i + 2}{2^{nk/2}}$
    \item $|a-b|\cdot \Pr[I(a,b,\et)] < \frac{2}{2^{n(k-2)/2}}$
\end{itemize}
\end{lemma}

\begin{proof}
When $n/k$ bits has energy allocated, each bit gets $nk/2$ energy and therefore the probability that the reader reads that bit wrongly is $\frac{1}{2^{nk/2}}$. Consider an arbitrary input array $C$. Let us first examine $\Pr[I(a,b)]$ for $b> a$ arbitrary elements in $C$. Consider the index of the leading different bit of $a$ and $b$ $i$, i.e. $FD(a,b) = i$. For each index $j > i$, the probability that the comparison is incorrect at bit $i$ is $< 1/2^{nk/2}$. If $i \ge n - n/k$, the probability that the comparison is incorrect before getting to bit $i$ is $< \frac{n-i}{2^{nk/2}}$. The probability that the comparison is incorrect at bit $i$ is $< \frac{1}{2^{nk}}$, and the probability that the comparison is incorrect after bit $i$ is $< 1/2^{nk/2}$. Overall, we have the probability that the comparison is incorrect is $< \frac{n-i+2}{2^{nk/2}}$ if $FD(a,b) = i \ge n - n/k$.

The second statement follows from the first one. We note that if the leading different bit of $a,b$ is $i$ then $|a-b| < 2^i$. Thus we have $|a-b|.\Pr[I(a,b,\pi_t)] < 2^i \cdot \frac{n-i+2}{2^{nk/2}}$. This expression is biggest when $i = n$, thus $|a-b|.\Pr[I(a,b,\pi_t)] < \frac{2}{2^{n(k-2)/2}}$ if the index of the leading different bit of $a,b$ is $\ge n-n/k$.
\end{proof}

\begin{lemma}\label{lemma: sorting2.2}
Given a random input array $C$ of $N$ $n$-bit numbers and two random index $a, b < N$. Then if $FD(C[a], C[b]) = i \ge n - n/k$ then $\Pr[Q(C[a],C[b], \et)]\cdot |C[a] - C[b]| < \frac{2.8 \log N}{2^{n(k-2)/2}}$.
\end{lemma}
\begin{proof}
We will use results from Lemma \ref{lemma: sorting2.1}. Let us again look at the 4 cases of pivot comparison. This time, we note that if $a < b < j$, $FD(a,j) \ge FD(a,b)$ (this is straightforward to check) and if $a < j < b$, either $FD(a,j)$ or $FD(b,j) = FD(a,b)$. We also note that if $FD(a,b) < n - n/k$, the comparison will be always be false since the bits after index $n - n/k$ always get flipped. 
\begin{itemize}
    \item Case 1: $C[a] < C[b]< C[j]$. In this case $Q(C[a],C[b], C[j], \ea)$ happens if only $I(C[a], C[j])$ happens. Therefore in this case 
    \begin{equation*}
        \begin{split}
            \Pr[Q(C[a],C[b], C[j], \et)] &= \Pr[I(C[a],C[j], \et)] \cdot(1-\Pr[I(C[b],C[j], \et)]) \\
            &< \Pr[I(C[a],C[j], \et)]
        \end{split}
    \end{equation*}
    If $FD(a,b) \ge n - n/k$ then $FD(a,j) \ge n - n/k$, thus $\Pr[I(C[a],C[j], \et)] < \frac{2}{2^{n(k-2)/2}}\cdot \frac{1}{C[j] - C[a]} < \frac{2}{2^{n(k-2)/2}}\cdot \frac{1}{C[b] - C[a]}$, thus  $\Pr[Q(C[a],C[b], C[j], \et)]\cdot|C[a] - C[b]| < \frac{2}{2^{n(k-2)/2}}$.
    \item Case 2: $C[j] < C[a] < C[b]$. In this case $Q(C[a],C[b],C[j], \ea)$ happens if only $I(C[b],C[j], \et)$ happens. Therefore again 
    \begin{equation*}
    \begin{split}
        \Pr[Q(C[a],C[b],C[j], \et)] & = \Pr[I(C[b],C[j], \et)].(1-\Pr[I(C[a],C[j], \et)])\\
        &< \Pr[I(C[b],C[j], \et)]
    \end{split}
    \end{equation*}
    Similar to the above case, in this case $\Pr[Q(C[a],C[b], C[j], \et)]\cdot|C[a] - C[b]| < \frac{2}{2^{n(k-2)/2}}$.
    \item Case 3: $C[a] < C[j] < C[b]$. In this case $Q(C[a],C[b],C[j], \ea)$ happens if both $I(C[b],C[j], \ea)$ and $I(C[a],C[j], \ea)$ happen. Therefore,
    \begin{equation*}
    \begin{split}
    \Pr[Q(C[a],C[b],C[j], \ea)] &= \Pr[I(C[b],C[j], \ea)].\Pr[I(C[a],C[j], \ea)]
    \end{split}
    \end{equation*}
    Now, note that if $C[a] < C[j] < C[b]$, either $FD(C[a],C[j])$ or $FD(C[b], C[j]) = FD(C[a], C[b])$. Assume it is $FD(C[a],C[j])$, then if $FD(C[a],C[b]) \ge n - n/k$ then so is $FD(C[a],C[j])$. Thus we have $\Pr[Q(C[a],C[b],C[j], \ea)] \le \Pr[I(C[a],C[j], \ea)] < \frac{n-i +2}{n^{nk/2}}$, and so $\Pr[Q(C[a],C[b],$ $C[j], \ea)]\cdot |C[a] - C[b]| < 2^i \cdot \frac{n-i +2}{n^{nk/2}} \le \frac{2}{2^{n(k-2)/2}}$ as proved in lemma \ref{lemma: sorting2.1}.
    \item Case 4: either $a = j$ or $b = j$. In this case, $\Pr[Q(C[a],C[b],C[j], \ea)] = \Pr[I(C[a],C[b], \ea)]$ and so $\Pr[Q(C[a],C[b],C[j], \ea)] \cdot |C[a] - C[b]| < \frac{2}{2^{n(k-2)/2}}$.
\end{itemize}
Therefore, if $FD(C[a], C[b]) = i \ge n - n/k$ then $\Pr[Q(C[a],C[b],C[j], \ea)] \cdot |C[a] - C[b]| < \frac{2}{2^{n(k-2)/2}}$. Thus, 
\begin{equation}
    \begin{split}
        \Pr[Q(C[a],C[b], \et)]\cdot |C[a] - C[b]| & \le \sum_j \Pr[Q(C[a],C[b], C[j], \et)]\cdot |C[a] - C[b]| \\
        &< 1.4 \log N \frac{2}{2^{n(k-2)/2}}
    \end{split}
\end{equation}

\end{proof}

Now we can prove theorem \ref{thm: truncate1}. The basic idea is that given two random numbers, the probability that their leading different bit is outside of the $\frac{n}{k}$ most significant bits is low. 

\begin{proof} [Proof of Theorem \ref{thm: truncate1}]
From Lemma \ref{lemma: sorting2.2}, we have for arbitrary input array $C$, $\Pr[Q(C[a],C[b], \et)]\cdot |C[a] - C[b]| < \frac{2.8 \log N}{2^{n(k-2)/2}}$ if $FD(C[a], C[b]) \ge n - n/k$.

On the other hand, for any pair of number $a,b$ the probability that $FD(a,b) < n - n/k$  is $\frac{1}{2^{n/k}}$. Therefore, for random $a,b$ from the uniform random distribution, 
\begin{equation}
    \Pr\left[|a-b|.\Pr[Q(a,b,\et)] > 2.8 \frac{\log N}{2^{n(k-2)/2}}\right] < \frac{1}{2^{n/k}}
\end{equation}

Thus, from the union bound,
\begin{equation} \label{eq7}
    \Pr\left[\wkt{C}(\et) = \sum_{a,b \in C}|a-b|\cdot \Pr[Q(a,b,\et)] = \Omega(\frac{N^2\log N}{2^{n(k-2)/2}}) \right] =  O(\frac{N^2/2}{2^{n/k}})
\end{equation}

On the other hand, from equation \ref{eq6} in the proof of Theorem \ref{thm: sorting1}, we know that for an input $C$ from the uniform random distribution:

\begin{equation} \label{eq8}
    \Pr[\wkt{C}(\eo) =  O(N\cdot 2^{n/6})] < \frac{N^2}{2^{(n-2)/3}}
\end{equation}
Therefore, from equations \ref{eq7} and \ref{eq8} and the union bound we have:
\begin{equation}
    \Pr\left[\frac{\wkt{C}(\eo)}{\wkt{C}(\et)} = O(\frac{2^{n(k-5/3 )/6}}{N\log N})\right] = O(\frac{N^2}{2^{\frac{n}{max(3,k)}}})
\end{equation}
and the theorem follows.
\end{proof}

We will use the above lemmas to prove Theorem \ref{thm: truncate2} also. 
\begin{proof} [Proof of Theorem \ref{thm: truncate2}]
From equation \ref{eq5} in the proof of theorem \ref{thm: sorting1} we know that $\E_C[\wkt{C}(\eo)] = \Omega(N\cdot 2^{n/2})$. We want to limit $\E_C[\wkt{C}(\et)]$ to be asymptotically $O(2^{n/h})$ where $h > 2$ so that the ratio $\frac{\E_C[\wkt{C}(\eo)]}{\E_C[\wkt{C}(\et)]}$ is exponential in $n$. From equation \ref{eq:wkt3}, we only have to limit $\E_{a \ne b}[|a-b|\cdot\Pr[Q(a,b, \et)]]$.

\begin{equation}
\begin{split}
    \E_{a \ne b}[|a- b|\Pr[Q(a,b, \pi_k)]] & = \Pr[FD(a,b) \ge n - n/k]\cdot\E_{a,b|FD(a,b) \ge n - n/k}\left[|a-b|\cdot\Pr[Q(a,b, \et)]\right]\\
    & + \Pr[FD(a,b) < n - n/k]\cdot\E_{a,b|FD(a,b) < n - n/k}\left[|a-b|\cdot\Pr[Q(a,b, \et)]\right] \\
    & <  \E_{a,b|FD(a,b) \ge n - n/k}\left[|a-b|\cdot\Pr[Q(a,b, \et)]\right] \\
    &+ \sum_{i < n - n/k}\Pr[FD(a,b) = i]\E_{a,b|FD(a,b) = i}|a-b|
\end{split}
\end{equation}

For the first sum, from Lemma \ref{lemma: sorting2.2} we have that $|a-b|\cdot\Pr[Q(a,b, \et)] < 1.4 \log N \frac{2}{2^{n(k-2)/2}} $ for any input array $C$ containing numbers $a,b$. It is clear that this sum is a constant for any $k > 2$.

For the second sum, we have
\begin{align*}
    \sum_{i < n - n/k}\Pr[FD(a,b) = i]\E_{a,b|FD(a,b) = i}|a-b| < \sum_{i < n- n/k}1/2^{n+1-i} \cdot 2^{i+1}\\
    = \sum_{i < n- n/k} 2^{2i - n} = \frac{2^{2(n- n/k) - n + 2}}{3}  = \frac{2^{n- 2n/k + 2}}{3}
\end{align*}

Therefore,
$$ \E_C[\wkt{C}(\et)] = \Omega(N^2\cdot 2^{n- 2n/k})$$

In order for this sum to be asymptotically $O(2^{n/2})$, $k$ must be $< 4$.
\end{proof}

\end{document}